\documentclass[sigplan]{acmart}

\usepackage{subcaption}
\usepackage{graphicx}
\usepackage{amsfonts}
\usepackage{multirow}
\usepackage{url}
\usepackage{listings}
\usepackage{tikz}
\usetikzlibrary{matrix}

\usetikzlibrary{shapes.geometric}
\tikzstyle{every node}=[font=\footnotesize]
\tikzstyle{etichetta}=[]
\tikzstyle{square}=[draw,outer sep=0pt,inner sep=-1.3em,regular polygon,regular polygon sides=4,,minimum size=13.5mm]
\tikzstyle{square2}=[square,dotted,black!50]
\pgfmathsetmacro{\ra}{0.95}
\usetikzlibrary{decorations.pathreplacing,calc}

\usepackage{listings}
\lstdefinelanguage{json}{
    showspaces=false,
    showtabs=false,
    breaklines=false,
    breakatwhitespace=false,
    basicstyle=\ttfamily\footnotesize,
    upquote=true,
    morestring=[b]",
    literate={♭}{$\flat$}1
}
\usepackage{algorithm}
\usepackage{algpseudocode}

\algnewcommand\algorithmicswitch{\textbf{switch}}
\algnewcommand\algorithmiccase{\textbf{case}}
\algnewcommand\algorithmicdefault{\textbf{default}}
\algnewcommand\algorithmicassert{\texttt{assert}}
\algnewcommand\Assert[1]{\State \algorithmicassert(#1)}%
\algdef{SE}[SWITCH]{Switch}{EndSwitch}[1]{\algorithmicswitch\ #1\ \algorithmicdo}{\algorithmicend\ \algorithmicswitch}%
\algdef{SE}[CASE]{Case}{EndCase}[1]{\algorithmiccase\ #1}{\algorithmicend\ \algorithmiccase}%
\algdef{SE}[DEFAULT]{Default}{EndCase}[1]{\algorithmicdefault\ #1}{\algorithmicend\ \algorithmiccase}%
\algtext*{EndSwitch}%
\algtext*{EndCase}%
\algnewcommand\algorithmicforeach{\textbf{for each}}
\algdef{S}[FOR]{ForEach}[1]{\algorithmicforeach\ #1\ \algorithmicdo}

\renewcommand\footnotetextcopyrightpermission[1]{}
\acmConference[]{}{}

\begin{document}

\title{A Composable CRDT Layer for Byzantine-Resilient Deterministic Reconstruction}

\author{Amos Brocco}
\email{amos.brocco@supsi.ch}
\orcid{0000-0002-0262-2044}
\affiliation{%
  \institution{University of Applied Sciences and Arts of Southern Switzerland}
  \city{Lugano}
  \country{Switzerland}
}

\begin{abstract}

Conflict-free Replicated Data Types (CRDTs) ensure Strong Eventual Consistency without coordination, but typically assume benign participants and rely on validation or exclusion to handle Byzantine behavior.

We address this problem through \emph{deterministic state reconstruction}: rather than deciding which updates are admissible, all accepted updates are incorporated, while only a subset contributes to the reconstructed state. We instantiate this approach in \emph{Melda}, a non-intrusive delta-state CRDT for JSON documents, and show that its reconstruction model guarantees convergence even under arbitrary update injection: adversarial updates are either structurally rejected or treated as inputs to the reconstruction process.

We formalize this model and prove that replicas deriving state from the same set of updates cannot diverge despite equivocation, omission, or message reordering. We further show that authentication, authorization, and confidentiality can be layered without affecting convergence.

Overall, this approach suggests that Byzantine tolerance can be achieved by decoupling update propagation from state derivation, allowing agreement on updates to be handled independently by external dissemination or consensus mechanisms.
\end{abstract}
\maketitle

\section{Introduction}

Conflict-free replicated data types (CRDTs) \cite{b1,b10,b32} provide a framework for achieving Strong Eventual Consistency (SEC) in distributed systems without coordination. They enable replicas to be updated independently, without the need for centralized or distributed synchronization, while guaranteeing that replicas that observe the same set of updates converge to the same state \cite{b0}. CRDTs span a wide spectrum of data structures, ranging from simple types such as counters, registers, and sets \cite{b10}, to more complex replicated structures \cite{b4,b13}.

CRDTs can be classified into three categories: operation-based, state-based, and delta-state based. Operation-based CRDTs \cite{b2} propagate individual updates and are well suited for high-frequency interactive applications. State-based CRDTs \cite{b10}, on the other hand, exchange entire states, simplifying correctness reasoning at the cost of increased communication overhead. To address this limitation, delta-state CRDTs ($\delta$-CRDTs) have been introduced \cite{b3,b14}, where replicas exchange small, idempotent delta mutations that can be safely applied multiple times, even over unreliable communication channels.

Classical CRDT designs assume benign participants and correct execution of the protocol. However, in open and decentralized environments, nodes may exhibit Byzantine behavior, generating malformed, conflicting, or malicious updates. As shown by Kleppmann~\cite{b23}, traditional CRDT designs do not guarantee convergence in the presence of such behavior unless additional mechanisms are introduced. In particular, validation-based approaches enforce that updates must satisfy deterministic conditions derived from their causal dependencies, ensuring that all correct replicas make consistent decisions about whether to accept or reject an update.

In this work, we explore an alternative approach that does not rely solely on validation to ensure convergence. Instead, we propose an approach based on deterministic state reconstruction built on \emph{Melda}, a delta-state CRDT designed for arbitrary JSON documents \cite{b20,b22}. Melda represents updates as immutable delta blocks organized in a causal structure, and reconstructs the state using a deterministic merge procedure. In contrast to other existing CRDT data structures, Melda is designed with practicality in mind: it operates as a non-intrusive CRDT overlay over arbitrary JSON data, requires no changes to application data models, and can be deployed over existing communication and storage systems.

Rather than deciding which updates are valid, we ensure that all admissible updates are interpreted deterministically. Our key observation is that, under suitable conditions, deterministic state reconstruction can guarantee convergence even in the presence of arbitrary update injection. In particular, if updates are causally structured, do not depend on hidden or replica-specific state, and are merged deterministically, Byzantine behavior cannot induce divergence among correct replicas, even when malicious updates are not explicitly rejected. In this setting, correctness arises not only from the admissibility of updates, but from their consistent interpretation.

This perspective also extends beyond classical CRDT designs, which focus on ensuring convergence under well-defined operations and do not explicitly account for adversarial update injection. In contrast, our model accommodates arbitrary updates—including malformed or malicious ones—and ensures convergence through deterministic interpretation rather than relying solely on the intended semantics of individual operations.

We further extend this model with optional pluggable and composable security mechanisms for authentication, authorization, and confidentiality. By introducing signature-based policies and optional encryption layers, we enable replicas to filter and restrict visibility of updates while preserving convergence within groups of replicas that share the same trust configuration.

Our main result can be summarized as follows. If two correct replicas observe the same set of updates, they deterministically derive the same state, even in the presence of Byzantine behavior. In particular, adversarial actions such as equivocation, reordering, or arbitrary update injection do not lead to divergence: updates are either consistently rejected or incorporated in a way that leads to the same outcome at all replicas.

Our approach does not introduce additional structural overhead compared to the underlying CRDT. Instead, it preserves the existing data representation and propagation mechanisms, and achieves Byzantine resilience by modifying only the state reconstruction process.

This guarantee relies on a simple structural requirement: validation, applicability, and state reconstruction are all deterministic and depend only on explicitly available data and declared dependencies. As a result, replicas that observe the same updates make the same decisions about which updates to accept, how to apply them, and how to interpret them.

Overall, this shows that convergence does not require replicas to agree on the validity of updates, but only on how a shared set of updates is interpreted.

This paper makes the following contributions:

\begin{itemize}
    \item We introduce a \emph{deterministic reconstruction model} for CRDTs, in which convergence follows from the deterministic interpretation of updates.

    \item We formalize convergence under Byzantine behavior as state derivation from a shared set of updates, and show that adversarial updates cannot induce divergence under deterministic validation and applicability.

    \item We characterize Byzantine behavior in this model, showing that equivocation and arbitrary update injection result in an all-or-nothing semantics: updates are either consistently ignored or incorporated.

    \item We show that authentication, authorization, and confidentiality can be layered as refinements of the validity predicate without affecting convergence.

    \item We instantiate the model with \emph{Melda}, a practical delta-state CRDT for JSON documents with a non-intrusive design and pluggable storage and transport layers.
\end{itemize}

The remainder of the paper presents related work, describes the Melda data model, formalizes its convergence guarantees under Byzantine behavior, and concludes with directions for future work.

\section{Related Work}
\label{sec:related}

\begin{figure*}[t]
\centering
\small
\begin{tabular}{lccc}
\textbf{Approach} & \textbf{Strategy} & \textbf{Byzantine Handling} & \textbf{Key Property} \\
\hline
Validation (Kleppmann) \cite{b23} & Accept / reject & Consistent validation & Agreement on validity \\
Equivocation-tolerant (Jacob) \cite{b24,b28} & Accept as concurrent & Absorb equivocation & No coordination \\
Blocklace \cite{b27} & Validate + detect & Detect and bound faults & Identity-based control \\
SRDT \cite{b26} & Centralized filtering & Restrict via policy & Trusted leader \\
Bounded CRDT (Baquero et al.) \cite{b33} & Cost bounding & Limit via PoW & Bounded adversarial impact \\
\textbf{Melda} & Deterministic reconstruction & Ignore via projection & Deterministic state \\
\end{tabular}
\caption{Comparison of approaches for handling Byzantine behavior in CRDT-based systems. Orthogonal mechanisms for security, such as encryption and authentication, are discussed separately.}
\label{fig:comparison}
\end{figure*}

Recent work has investigated the behavior of CRDTs in the presence of Byzantine nodes, highlighting that classical guarantees do not automatically extend to adversarial environments. Kleppmann~\cite{b23} shows that standard CRDT algorithms may diverge under Byzantine behavior, even when eventual delivery is satisfied. In particular, divergence arises when replicas make inconsistent decisions about the validity of updates. To address this issue, Kleppmann proposes a validation-based approach in which updates are accepted only if they satisfy deterministic predicates derived from their causal dependencies. Convergence is preserved provided that all replicas evaluate these predicates consistently.

Jacob et al.~\cite{b24} analyze Byzantine behavior in CRDTs through the notion of \textit{equivocation tolerance}. They show that, at the CRDT layer, Byzantine behavior can be reduced to two classes of faults: \textit{equivocation} and \textit{omission}. Equivocation refers to the ability of a faulty replica to present different, inconsistent versions of its updates to different peers, while omission denotes the selective withholding or loss of messages, preventing some updates from being delivered.

In equivocation-tolerant designs, equivocated updates are treated as causally independent and therefore as concurrent operations. As a result, such behavior does not violate convergence, but is instead absorbed into the CRDT semantics.

Under the assumption that malformed updates are filtered by deterministic validity checks, divergence does not arise from invalid data generation itself, but from inconsistent dissemination of updates across replicas. This perspective highlights that Byzantine faults at the CRDT layer primarily affect dissemination rather than the semantics of individual updates.

Building on these insights, in ~\cite{b28} a more general formulation of Byzantine-tolerant CRDTs based on \textit{extend-only directed posets} (EDPs) is proposed. In this model, replica states are represented as monotonically growing partially ordered sets, and convergence is achieved via set union as a join operation. This abstract formulation emphasizes that Byzantine tolerance arises from monotonic growth and deterministic merging, rather than from coordination or consensus.

These ideas extend naturally to DAG-based CRDT designs, where updates are organized as hash-linked nodes encoding causal dependencies. In such systems, equivocated updates are treated as concurrent operations and incorporated into the structure without requiring detection or prevention. While this approach guarantees Strong Eventual Consistency, it allows Byzantine nodes to introduce an unbounded number of updates, potentially polluting the replicated state and affecting system performance. 

This limitation motivates approaches that augment equivocation tolerance with mechanisms to detect or constrain adversarial behavior, while preserving coordination-free convergence.

One representative example is Blocklace~\cite{b27}, which extends this line of work by introducing a partially ordered DAG in which each block contains signed references to predecessors. It can be viewed as a generalized delta-state CRDT, where blocks act as delta fragments and are merged via set union under causal closure conditions. In contrast to purely equivocation-tolerant approaches, Blocklace enforces structural and semantic validity predicates and leverages cryptographic identities to detect equivocation and eventually exclude Byzantine nodes. This allows bounding the impact of adversarial behavior to a finite prefix of the computation, at the cost of increased system complexity and stronger assumptions on identity management.

Recent work has also explored how to bound the impact of Byzantine behavior in CRDT systems, rather than preventing or detecting it. In particular, Baquero et al.~\cite{b33} introduce the notion of \emph{bounded Byzantine CRDTs}, where each update is associated with a computational cost proportional to its semantic impact. By leveraging proof-of-work mechanisms, their approach ensures that the cumulative effect of adversarial updates is bounded by the computational resources available to the attacker. This model complements prior work on Byzantine-tolerant CRDTs: while cryptographic validation ensures structural correctness, bounded-cost mechanisms limit the rate and magnitude of adversarial influence.

In contrast to our approach, which tolerates arbitrary sets of updates and derives correctness through deterministic state reconstruction, bounded Byzantine CRDTs constrain the admissible behavior of replicas by imposing resource-based limits on update generation. As a result, their model preserves convergence while bounding adversarial impact, whereas our work focuses on ensuring deterministic state derivation independently of update admissibility. These approaches are orthogonal and can be combined: resource-bounded update generation can reduce the volume of adversarial inputs, while deterministic reconstruction ensures that only structurally valid dependencies contribute to the replicated state.

A complementary line of research focuses on integrating security properties such as confidentiality, integrity, and access control into replicated data systems. These works explore how replicated data types can enforce security guarantees directly at the data layer, typically through cryptographic techniques or policy-aware replication mechanisms. While orthogonal to convergence itself, such mechanisms are often required in practical deployments and can be composed with replication protocols. In this context, it is relevant to consider how these security concerns interact with deterministic state reconstruction, as explored in the following works.

Jannes et al.~\cite{b25} propose a secure replication protocol for state-based CRDTs that combines fine-grained encryption with Byzantine fault tolerance. Their approach ensures confidentiality, integrity, and attributability of updates, while preserving Strong Eventual Consistency even in the presence of Byzantine replicas. Access to data is controlled cryptographically, and only authorized users can read or modify the encrypted state.

Renaux et al.~\cite{b26} propose Secure RDTs (SRDTs), a data type that enforces role-based access control over replicated JSON data. Their approach relies on a trusted central leader that defines and enforces a global security policy. The leader provides each client with a projection of both the data and the policy tailored to its role, ensuring that replicas only contain authorized data and that unauthorized updates are rejected. This design prevents replicated data leaks and unauthorized updates while preserving eventual consistency, but introduces a centralized component.

Our work differs from these approaches by shifting the focus from validation, equivocation handling, or centralized enforcement to deterministic state reconstruction. In our approach, all updates — including malformed or adversarial ones — are propagated and stored, but the reconstructed state depends only on a deterministically defined subset of structurally valid dependencies. Rather than enforcing correctness by rejecting updates or filtering them during dissemination, we separate propagation from application: correctness is obtained as a deterministic projection over the accumulated set of updates.

In this sense, our approach can be seen as a delta-state counterpart to DAG-based CRDTs. While Blocklace introduces validation and exclusion mechanisms, SRDTs rely on centralized filtering, we show that convergence does not require global agreement on validity predicates, and can instead be ensured at the data layer through deterministic reconstruction. This perspective shifts correctness from update admissibility or invariant preservation to state derivation.

Finally, our design decouples convergence from security concerns. Authentication and authorization are expressed as refinements of the validity predicate, allowing different replicas to enforce different trust configurations while preserving deterministic convergence within each trust domain.

Figure~\ref{fig:comparison} summarizes the different strategies for handling Byzantine behavior and highlights the distinct design point explored in this work.

\section{Melda}
\label{sec:melda}
In this section, we present an overview of Melda\footnote{\url{https://github.com/slashdotted/libmelda}} with a focus on the aspects that are relevant to its Byzantine resilience. A comprehensive description of the data structure and its design rationale is provided in our previous work \cite{b20,b22}. 

\subsection{System model}

Melda is a delta-state replication framework that aims to guarantee eventual consistency in a decentralized and asynchronous environment. Each replica evolves independently by generating local updates, which are encoded as deltas and incorporated into a grow-only structure. Replicas exchange these deltas asynchronously, and incorporate those whose causal dependencies and integrity constraints are satisfied.
Unlike operation-based approaches, Melda does not rely on the delivery of individual operations in a specific order. Instead, convergence is achieved through the accumulation and reconciliation of delta states, which can be transmitted, stored, and validated independently. This design enables replicas to operate under network partitions, tolerate message reordering, and progressively converge toward a consistent state without requiring coordination, global ordering, or trusted infrastructure.

\subsection{Application data model}

Unlike CRDT frameworks such as Automerge\footnote{\url{https://automerge.org/}}, which introduce their own data model and require applications to adopt specialized data structures, Melda operates directly on the application's existing data model. Rather than replacing or redefining the structure of application data, Melda processes a JSON serialization of that model and detects changes at the serialization level. The only requirement is a deterministic serialization and deserialization to and from JSON. This allows Melda to remain non-intrusive while enabling collaborative behavior.

\subsection{Object store and versioning}

Melda is based on a grow-only collection $C$ of JSON objects, replicated across multiple sites and concurrently updated by each participant. Objects are atomic and immutable: even partial modifications result in the creation of a new version, which is appended to the collection. Deletions are recorded using tombstones.

An object $o \in C$ is a set of name-value pairs and is uniquely identified by a stable identifier $id_o$. Each version $x$ of an object is associated with a cryptographic hash $H(x)$, computed from its content (excluding the object identifier). This allows replicas to efficiently compare versions and retrieve their associated values.

Each version is also associated with a \textit{revision string}. Due to concurrent updates, multiple revision sequences may coexist for the same object, forming a revision tree. A deterministic algorithm selects the \textit{winning revision}, typically corresponding to the longest chain. The state of the CRDT is defined by the collection $C$ together with all revision trees.

\subsection{Update generation}

To perform an update, Melda processes a JSON serialization of the application's data model and decomposes it into a collection of objects. Using a reversible transformation algorithm, nested structures are flattened into the collection $C$, and references are introduced to preserve the original structure. The resulting top-level object is called the \textit{root}.
Objects are then compared with the current state, and differences generate new revisions. For large arrays, compact patches are generated relative to previous revisions to reduce storage overhead.

\subsection{State reconstruction}

To reconstruct the document from a given state, the object references in the root object (representing the starting $value$) are recursively resolved by replacing them with the values of their winning revisions.
In the presence of conflicts caused by concurrent modifications, Melda deterministically selects one of the competing revisions as the winner. However, the system allows the application layer to override this choice, enabling domain-specific conflict resolution policies.
Only revisions that are causally reachable from the root object are considered during reconstruction. Revisions that are not connected to any valid dependency chain rooted in the initial state are ignored, ensuring that only causally well-founded branches contribute to the final state.

\subsection{Array conflicts and merging}

Conflicts within arrays of objects are handled through an optional non-destructive array merging strategy. In the default model, arrays are fields of objects, and concurrent updates lead to multiple conflicting revisions of the enclosing object. Selecting a single revision may discard concurrent changes.
To mitigate this, Melda optionally elevates arrays of objects to first-class entities. These arrays are extracted and managed as independent objects in the collection $C$, with their own versioning history. Conflict resolution is then performed through a deterministic merging algorithm that integrates concurrent insertions and removals.
For example, given an array $[A, B, C]$, one replica may append \textbf{D}, producing $[A, B, C, \textbf{D}]$, while another inserts \textbf{E} between \textbf{A} and \textbf{B}, producing $[A, \textbf{E}, B, C]$. Without lifting, one revision would be selected and the other discarded. With array lifting enabled, both updates are merged into $[A, \textbf{E}, B, C, \textbf{D}]$.
This mechanism applies only to arrays of objects and does not alter the underlying CRDT semantics. It is a deterministic reconciliation step that reduces application-level complexity.

\subsection{Delta propagation, adapters, and validation}

Updates in Melda are propagated as \textit{delta blocks}. Each delta block encodes one or more revisions and references a set of prior delta blocks, called \textit{anchors}, which represent its causal dependencies. Each delta block is uniquely identified by the cryptographic hash of its serialized content, ensuring immutability and preventing equivocation.

To enforce a consistent causal structure, delta blocks are also assigned a monotonically increasing index. The index of a delta block is defined as $N = \max(\text{index of its anchors}) + 1$, ensuring that every block is causally ordered after its dependencies. This monotonic assignment guarantees that the dependency graph is acyclic and induces a partial order over updates. For convenience, a delta block may be identified by a composite name of the form $N$-\textit{digest}, where $N$ is the index and \textit{digest} is the cryptographic hash.

Actual JSON values associated with revisions are stored in \textit{data packs}, which are also content-addressed by hash. This allows replicas to verify integrity independently of the source.

\paragraph{Adapters.}
Melda interacts with storage and transport layers through an abstraction called an \textit{adapter}. An adapter is responsible for persisting and retrieving delta blocks and data packs, but does not alter the semantics of the CRDT. Adapters can be composed in a pipeline fashion, allowing additional processing layers—such as encryption, compression, or security policies—to be transparently applied. This modular design enables Melda to remain agnostic to the underlying storage and communication mechanisms, while preserving deterministic behavior.

An important consequence of this design is that Melda naturally supports \textit{opportunistic synchronization}. Since updates are materialized as immutable delta blocks and data packs, their dissemination can be delegated to external storage or synchronization systems (e.g., cloud file synchronization, shared filesystems, or version control systems). Replicas can independently commit changes locally and later incorporate remote updates when they become available, without requiring a dedicated communication protocol or persistent connectivity. This approach decouples consistency from the communication layer, enabling asynchronous, offline-first workflows while maintaining deterministic convergence.

\paragraph{Validation.}
Upon receiving data, replicas validate delta blocks and data packs before incorporation. A delta block is accepted only if:
\begin{itemize}
    \item it is a well-formed data structure;
    \item its hash matches the digest contained in its identifier;
    \item all referenced anchors (causal dependencies) are available and valid;
    \item its index satisfies the monotonicity rule described above with respect to its anchors;
    \item all required data packs are available and pass hash verification;
    \item optionally, if authentication is enabled (see Section \ref{sec:melda-sec}), the delta block carries a valid signature.
\end{itemize}

Data packs are similarly accepted only if their hash matches the expected content. Invalid or tampered data is therefore deterministically rejected by all replicas.

\paragraph{Applicability.}
A delta block is applied if and only if all its anchors have been applied and the required data is available. Otherwise, it is temporarily kept in a pending state until the missing dependencies are satisfied. This guarantees causal consistency and ensures that updates are incorporated according to the partial order induced by the dependency graph.
Crucially, both validation and applicability depend only on the content of the delta block, its hash, and its explicitly declared dependencies. They do not depend on hidden or replica-specific state. As a consequence, all correct replicas receiving the same set of delta blocks will make identical decisions about which updates are accepted and when they are applied.

\paragraph{Reconciliation.}
Since both delta blocks and data packs are content-addressed, replicas reconcile their state by exchanging hashes and retrieving missing data. This mechanism tolerates message reordering, duplication, and dissemination through untrusted intermediaries, while ensuring eventual convergence. Overall, the causal dependency graph induced by delta blocks constitutes the sole source of truth for validation and application decisions.

\subsection{Security extension: signing, authorization, and encryption}
\label{sec:melda-sec}
Melda can be extended with a pluggable security layer, referred to as \textit{melda-sec}\footnote{\url{https://github.com/slashdotted/libmelda-sec}}, which introduces cryptographic authentication, authorization, and optional confidentiality while preserving the underlying CRDT semantics.

This extension is implemented as an adapter that wraps an existing Melda adapter, following a decorator pattern. As a result, it composes naturally with other adapters in the processing pipeline (e.g., encryption or compression), and operates transparently with respect to the propagation and validation mechanisms described above.

\paragraph{Cryptographic authentication.}
Each delta block may be signed using a private key associated with the originating replica. Upon reception, signatures are verified against a set of trusted public keys maintained in a \textit{keystore}. This ensures that only updates produced by authorized identities can be considered valid.

\paragraph{Policy-based authorization.}
In addition to signature verification, melda-sec introduces a \textit{policy engine} that determines which identities are allowed to modify specific objects. Policies are evaluated based on object identifiers and may optionally refer to roles associated with public keys. During state reconstruction, only revisions that satisfy both signature verification and policy constraints are treated as valid.

Importantly, all data—including invalid or unauthorized updates—is still persisted and propagated through the system. However, such updates are deterministically ignored when reconstructing the state. This ensures that authorization does not interfere with data availability or synchronization.

\paragraph{Encryption.}
Melda can also be composed with an encryption adapter, which encrypts delta blocks and data packs prior to storage or transmission. Decryption is performed locally upon retrieval. Since encryption operates at the adapter level, it does not modify the structure of delta blocks or their dependencies, and therefore does not affect the consistency guarantees of the system.

\paragraph{Design properties.}
The security extension preserves the core properties required for convergence:

\begin{itemize}
    \item \textit{Deterministic validation:} signature verification and policy evaluation depend only on the delta block content and the local trust configuration;
    \item \textit{Separation of concerns:} propagation, validation, and authorization are decoupled, allowing invalid updates to be filtered without affecting dissemination;
    \item \textit{Compatibility with delta-state replication:} validation is local and does not require coordination or global agreement.
\end{itemize}

The extension refines the validity predicate by adding authentication and policy constraints, but does not introduce consensus or enforce a globally shared authorization policy. Replicas with the same trust configuration will converge to the same state, while different configurations may yield different views over the same underlying set of updates.

\section{Byzantine Resilience}
\label{sec:bft}

In this section, we show that Melda is resilient to Byzantine behavior, in the sense that malformed, inconsistent, or adversarial updates cannot cause divergence among correct replicas. In particular, we prove that any two replicas that derive their state from the same set of delta blocks reach the same result, regardless of the presence of Byzantine nodes.

This notion of resilience is closely related to Strong Eventual Consistency (SEC)\cite{b32}, which requires that replicas that have observed the same set of updates eventually converge to the same state. Our result can be seen as establishing SEC under Byzantine behavior, given deterministic validation, applicability, and state reconstruction.

We characterize Byzantine resilience as the composition of deterministic validation, applicability, and reconstruction over a shared set of updates.

\subsection{Adversary model}

We consider a decentralized and asynchronous system in which replicas communicate over an unreliable network, and assume that correct replicas form a connected component, such that any update disseminated by a correct node is eventually received by all other correct replicas. The adversary may control an arbitrary number of nodes and is allowed to:

\begin{itemize}
    \item generate arbitrary delta blocks and data packs;
    \item modify, reorder, delay, or drop messages;
    \item attempt to inject malformed, inconsistent, or adversarial updates.
\end{itemize}

We assume that standard cryptographic primitives hold, namely:

\begin{itemize}
    \item hash functions are collision-resistant;
    \item if authentication is enabled, digital signatures cannot be forged.
\end{itemize}

Correct replicas follow the validation and application rules described in Section~\ref{sec:melda}.

Malformed updates are excluded by local validation, while inconsistencies introduced by adversarial behavior arise from the dissemination of updates across replicas. Byzantine nodes may inject arbitrary updates or present different updates to different replicas, which are incorporated as distinct delta blocks.

We do not assume any bound on the number of updates that a Byzantine replica may generate, allowing adversarial nodes to flood the system with arbitrary delta blocks. While such behavior can impact system performance or storage size, it does not affect convergence, which depends solely on the consistent interpretation of the accumulated set of updates.

\subsection{Validity predicates}
Validity predicates define the admission criteria for updates. They determine which delta blocks can be considered for inclusion in the replicated state, and provide the basis for reasoning about consistent behavior under adversarial conditions by ensuring that all replicas agree on which updates are admissible.

Let $u$ be a delta block. We denote by $\mathrm{valid}(u)$ the base validity predicate defined in Section~\ref{sec:melda}.

When the security extension (Section~\ref{sec:melda-sec}) is enabled, the validity predicate is refined into $\mathrm{valid}_{sec}(u)$, defined as:

\[
\mathrm{valid}_{sec}(u) = \mathrm{valid}(u) \wedge \mathrm{sig}(u) \wedge \mathrm{policy}(u)
\]

where:
\begin{itemize}
    \item $\mathrm{sig}(u)$ denotes successful signature verification against a trusted public key;
    \item $\mathrm{policy}(u)$ denotes compliance with local authorization rules.
\end{itemize}

Thus, the security extension restricts the set of accepted delta blocks without altering the underlying propagation mechanism.

\subsection{Deterministic validation}
Deterministic validation ensures that all replicas make identical decisions about which updates are accepted. This prevents divergence caused by inconsistent acceptance of updates under Byzantine behavior.

\begin{lemma}[Deterministic validation]
\label{lem:deterministic-validation}
For any delta block $u$, all correct replicas with the same configuration evaluate $\mathrm{valid}(u)$ (or $\mathrm{valid}_{sec}(u)$, when the security extension is enabled) identically.
\end{lemma}

\begin{proof}[Sketch]
The predicate $\mathrm{valid}(u)$ depends only on:
(i) the structure of $u$,
(ii) its hash,
(iii) its declared anchors, and
(iv) the availability and integrity of referenced data packs.

When the security extension is enabled, $\mathrm{valid}_{sec}(u)$ additionally depends on:
(v) signature verification, and
(vi) policy evaluation.

All these checks are deterministic functions of the delta block content and its explicitly referenced dependencies. Signature verification is deterministic under the cryptographic assumptions, and policy evaluation depends only on the local configuration. Therefore, replicas with identical configurations reach identical validity decisions.
\end{proof}

\subsection{Causal applicability}
Deterministic applicability guarantees that replicas not only agree on which updates are valid, but also on when they can be applied. This ensures that causal ordering is interpreted consistently across replicas, preventing divergence due to differences in application timing.

Let $A$ be the set of delta blocks already applied at a replica.

\begin{lemma}[Deterministic applicability]
\label{lem:deterministic-applicability}
For any delta block $u$, all correct replicas decide consistently whether $u$ is applicable once they have received the same set of delta blocks.
\end{lemma}

\begin{proof}[Sketch]
A delta block $u$ is applicable if and only if:
(i) it satisfies the validity predicate, and
(ii) all its anchors have already been applied.

Since the set of anchors is explicitly declared, and validity decisions are deterministic (Lemma~\ref{lem:deterministic-validation}), applicability depends only on the shared set of delta blocks. Therefore, replicas that have received the same set make identical applicability decisions.
\end{proof}

\subsection{Acyclic dependency graph}
Acyclicity ensures that the dependency relation is well-founded, allowing validation to be defined as a finite recursive process over dependencies. This prevents adversarial updates from introducing cycles that would otherwise make validation non-terminating or undefined.

\begin{lemma}[Acyclicity]
\label{lem:acyclic}
The dependency graph induced by delta blocks is a directed acyclic graph.
\end{lemma}

\begin{proof}[Sketch]
Each delta block is assigned an index defined as $N = \max(\text{indices of its anchors}) + 1$. Therefore, every dependency edge goes from a node with strictly lower index to one with strictly higher index. This strictly monotonic ordering prevents cycles.
\end{proof}

\subsection{Deterministic state}
Deterministic state reconstruction ensures that the final application state depends only on the set of accepted updates, rather than on the order in which they are applied. This eliminates a primary source of nondeterminism and guarantees convergence once a common set of updates is established.

\begin{lemma}[Deterministic state]
\label{lem:deterministic-state}
Let $U$ be a set of accepted delta blocks. The reconstructed state $\mathrm{state}(U)$ is uniquely determined.
\end{lemma}

\begin{proof}[Sketch]
Since the dependency graph is acyclic (Lemma~\ref{lem:acyclic}), the set $U$ admits a topological ordering. Since validation ensures that all dependencies of accepted delta blocks are satisfied, the set of updates is causally closed. As a result, delta blocks can be applied in any order. The state is then fully determined by the reconstruction process.

Since:
(i) applicability is deterministic (Lemma~\ref{lem:deterministic-applicability}),
(ii) revision trees are constructed deterministically,
(iii) winning revisions are selected deterministically, and
(iv) the merge of delta states is idempotent, associative, and commutative,
the resulting state is uniquely determined by $U$, independently of the order of application.

Only revisions that are causally reachable from the root contribute to the reconstruction, ensuring that semantically irrelevant or disconnected branches are consistently ignored across replicas.
This property implies that neither message omission nor equivocation can affect the final state beyond the selection of $U$: omission determines which updates are included in $U$, while equivocation can only introduce additional elements into $U$, both of which are handled deterministically.
\end{proof}

\subsection{Consistency of trust configuration}
Consistency of trust configuration is required to ensure that replicas sharing the same security assumptions evaluate updates in the same way. This allows different trust domains to coexist, while preserving convergence within each domain.

\begin{lemma}[Consistency of trust configuration]
\label{lem:trust-consistency}
If two replicas have identical sets of trusted public keys and identical policy rules, then they evaluate $\mathrm{valid}_{sec}(u)$ identically for all delta blocks $u$.
\end{lemma}

\begin{proof}[Sketch]
Signature verification depends only on the public keys available in the local keystore, and policy evaluation depends only on the configured rules. If these configurations are identical, then the evaluation of $\mathrm{valid}_{sec}(u)$ is identical across replicas.
\end{proof}

\subsection{Byzantine resilience}
Taken together, the previous properties ensure that acceptance, application, and interpretation of updates are all deterministic functions of the observed data. This allows convergence to be reduced to agreement on the set of updates, independently of adversarial behavior.

\begin{theorem}[Byzantine resilience]
\label{thm:byzantine-resilience}
Let $p$ and $q$ be two correct replicas that have incorporated the same set $U$ of delta blocks, and such that validation, applicability, and reconstruction are deterministic. If $p$ and $q$ share the same configuration (i.e., they evaluate $\mathrm{valid}$, respectively $\mathrm{valid}_{sec}$, identically), then:
\[
\mathrm{state}_p(U) = \mathrm{state}_q(U)
\]
even in the presence of Byzantine nodes.
\end{theorem}

\begin{proof}[Sketch]
Consider Byzantine behavior expressed through omission and equivocation.

\begin{itemize}
    \item \textbf{Omission.} Missing or delayed messages affect only which delta blocks are included in the set $U$.

    \item \textbf{Equivocation.} A Byzantine replica may present different delta blocks to different replicas.
    
    However:
    \begin{itemize}
        \item if equivocated updates are invalid, they are rejected by all replicas (Lemma~\ref{lem:deterministic-validation});
        \item if they are valid, they are treated as distinct delta blocks and incorporated into $U$.
    \end{itemize}
    Since state reconstruction is deterministic (Lemma~\ref{lem:deterministic-state}), these additional updates are interpreted consistently by all replicas that observe them.
\end{itemize}

Therefore, Byzantine behavior cannot cause inconsistent acceptance or interpretation of updates. Each delta block is either:
(i) rejected by all replicas,
(ii) incorporated into $U$ and ignored by some configurations, or
(iii) incorporated into $U$ and deterministically applied by all replicas sharing the same configuration.

As a consequence, two replicas that observe the same set $U$ derive the same state, regardless of adversarial behavior.

Hence:
\[
\mathrm{state}_p(U) = \mathrm{state}_q(U)
\]
\end{proof}

\subsection{Discussion}

When the security extension is enabled, convergence depends not only on the set of received delta blocks, but also on the trust configuration (i.e., the set of trusted public keys and the policy rules).

Replicas with identical configurations will converge to the same state, while replicas with different configurations may derive different subsets of accepted updates.

In this sense, melda-sec refines the validity predicate by introducing authentication and authorization constraints. This allows Byzantine updates to be not only harmless with respect to convergence, but also selectively excluded from the reconstructed state based on application-defined trust policies.

\subsection{Toward Full Byzantine Fault Tolerance}

The results presented in this work establish convergence under Byzantine behavior for any given set of accepted updates. However, they do not guarantee that all correct replicas will necessarily observe the same set of updates. In particular, omission and equivocation may lead to different replicas incorporating different subsets of delta blocks before reconciliation.

This observation suggests a modular decomposition of Byzantine fault tolerance. On one hand, dissemination and agreement mechanisms are responsible for ensuring that replicas eventually observe the same set of updates. On the other hand, deterministic state reconstruction ensures that all replicas derive the same state from that set. In this perspective, convergence is achieved as the composition of two orthogonal properties: agreement on the set of updates, and deterministic interpretation of those updates.

Different system designs can instantiate this decomposition. For instance, a centralized storage layer trivially enforces agreement at the cost of decentralization. Gossip-based dissemination ensures eventual propagation of updates, but does not guarantee agreement in adversarial settings. Stronger guarantees can be obtained by integrating a consensus mechanism that establishes quorum-based acceptance of delta blocks, ensuring that all correct replicas converge toward the same set of updates.

Importantly, Melda is compatible with all these approaches, as it does not impose constraints on the dissemination layer. This makes it possible to compose deterministic reconstruction with different communication and agreement strategies, ranging from fully decentralized to strongly coordinated systems.

These observations suggest that deterministic reconstruction provides a reusable foundation for Byzantine-resilient data systems, onto which stronger guarantees—such as agreement and fault containment—can be layered independently.

This separation highlights that deterministic reconstruction is orthogonal to agreement: it ensures that replicas cannot diverge once a common set of updates is established, independently of how that set is obtained.

\section{Conclusion}
\label{sec:conclusion}

In this paper, we explored an alternative perspective on Byzantine-tolerant CRDT design based on \emph{deterministic state reconstruction}. Rather than relying solely on agreement over update admissibility, we showed that convergence can be achieved by ensuring that propagated updates are interpreted deterministically.

We formalized this model and proved that divergence cannot occur among correct replicas that observe the same set of updates, even in the presence of Byzantine behavior. Adversarial actions such as equivocation, message reordering, or arbitrary update injection do not lead to inconsistent states: updates are either rejected due to structural invalidity or incorporated as deterministic inputs into the reconstruction process. This shifts correctness from update admissibility alone to the combination of validation and deterministic state derivation.

We demonstrated that this model is realizable through Melda, a delta-state CRDT for JSON documents that operates directly on application data without requiring intrusive data models or centralized coordination. By structuring updates as causally ordered delta blocks and applying deterministic reconstruction, Melda naturally tolerates reordering, duplication, and adversarial dissemination patterns.

We further showed that security mechanisms such as authentication, authorization, and confidentiality can be integrated as refinements of the validity predicate without affecting convergence: replicas with identical trust configurations converge to the same state, while different configurations yield internally consistent views over the same updates.

Overall, our work highlights a complementary design point in which robustness to adversarial behavior is achieved not only through validation or exclusion, but also through deterministic interpretation of accumulated updates.

Future work includes integrating peer-to-peer dissemination adapters (e.g., gossip-based protocols), further evaluating role-based authorization in decentralized settings, extending the Melda-sec model with richer cryptographic primitives, and studying more expressive reconstruction strategies and performance under adversarial workloads.

Future work also includes exploring mechanisms for key revocation and trust evolution, such as combining blacklist-based exclusion with content-based whitelisting to preserve previously accepted updates without breaking causal dependencies.

\bibliographystyle{acm}
\bibliography{melda-bft}

\end{document}